\documentclass[journal,twocolumn]{IEEEtran}

\IEEEoverridecommandlockouts
\usepackage{epsfig,color,cite}
\usepackage{amsmath}
\usepackage{algorithm}
\usepackage{algpseudocode}% 
\usepackage{amsthm}
\usepackage{mathrsfs}
\usepackage[center]{caption}
\usepackage{amsfonts}
\usepackage{amssymb}
\usepackage{mathrsfs} 
\usepackage{euscript}
\usepackage{graphicx}
\usepackage{epstopdf}
\usepackage{multirow}
\graphicspath{ {./images/} }
\newtheorem{lemm}{Lemma}

\begin{document}
	\title{Delay Minimization  in Sliced Multi-Cell Mobile Edge Computing (MEC) Systems\\
{\footnotesize \textsuperscript{}}
\vspace{-10mm}
\thanks{ S. Zarandi and H.~Tabassum are with the Lassonde School of Engineering at York University, Canada (e-mail: shz@yorku.ca, hina@eecs.yorku.ca). This work is supported  by the Discovery Grant from the Natural Sciences and Engineering Research Council of Canada.}
}
\vspace{-10mm}
\author{
\IEEEauthorblockN{Sheyda Zarandi and Hina Tabassum, \textit{Senior Member, IEEE}} 
	\vspace{-10mm}
}
	\raggedbottom
	\maketitle
	\begin{abstract}
Here, we consider the problem of jointly optimizing users' offloading decisions, communication and computing resource allocation in a  sliced multi-cell mobile edge computing (MEC) network. We minimize the weighted sum of the gap between the observed delay at each slice and its corresponding delay requirement, where weights set the priority of each slice. Fractional form of the objective function, discrete subchannel allocation, considered partial offloading, and the interference incorporated in the rate function, make the considered problem a complex mixed integer non-linear programming problem. Thus, we decompose the original problem into two sub-problems: (i) offloading decision-making and (ii) joint computation resource, subchannel, and power allocation. We solve the first sub-problem optimally and for the second sub-problem, leveraging on novel tools from fractional programming and Augmented Lagrangian method, we propose an efficient algorithm whose computational complexity is proved to be polynomial. Using alternating optimization, we solve these two sub-problems iteratively until convergence is obtained. Simulation results demonstrate the convergence of our proposed algorithm and its effectiveness compared to existing schemes. 

	\end{abstract}
\vspace{-5mm}
	\section{Introduction}
	Network slicing is an indispensable technique to support heterogeneous services in fifth generation (5G)  networks \cite{b1}. Using network slicing,  multiple logical network slices can be created on a common physical infrastructure. Each slice can be tailored to a specific application with distinct Quality-of-Service (QoS) requirement. On another note, resource-intensive and latency sensitive services necessitate mobile edge computing (MEC) that brings computational resources to the Radio Access Network (RAN) edge. Thus, users would use both RAN and computation resources to offload and process their tasks at the MEC servers.  On the other hand, in a sliced network, resources are restricted for each slice based on a service level agreement (SLA) with infrastructure provider (InP). Subsequently, joint optimization of RAN resources (e.g., subchannel and power) and computation resources (e.g., CPU cycles of MEC servers) with optimal computation offloading in a sliced network becomes imperative.
% This isolation of resources and strict limitations set on resource utilization of each slice turns joint RAN and computation resource allocation (RA) imperative in addressing the distinct latency requirements of diverse B5G services.
	
Recently, the problem of delay minimization in a multi-cell MEC network was solved through communication and computation resource allocations (RAs) \textit{without network slicing} \cite{b4,b5, b14}.  However, in all these works, the interference was either ignored \cite{b4,b5} or simplified  \cite{b14}. Also, in \cite{b4}, offloading decisions were not optimized, \cite{b5} did not consider RAN RA, and \cite{b14} considered a binary offloading scheme. 
% In \cite{b16}, a joint RAN and computation RA was considered, however, again the offloading decision was binary, the objective function was related to traffic management, and the method proposed for dealing with interference could not be generalized to scenarios where successive interference cancellation is not employed.
% % 	In \cite{b12,b13,b15,b17,b18},  non-orthogonal multiple access (NOMA)-enabled MEC system was considered assuming a single ES. 
% In \cite{b4} and \cite{b5}, the computational offloading decision was binary and the interference was ignored.
% % 	\textcolor{red}{In \cite{b14,b16}, multiple MEC servers exist in the network.} 
% 	In \cite{b14}, RAN and computation resources were considered; 

	A handful of research studies considered RA in sliced cellular networks \cite{b1,b3, b9, b10, b2,last}. In \cite{b1}, the authors minimized a weighted combination of energy consumption and delay through subchannel and computation RA. This work considered  two slices on a single base station (BS) with no interference. In \cite{b3}, the authors  minimized delay through computation RA, considering multiple BSs, and in \cite{b9}, the authors maximized the offloaded workload that can be supported in a given time at each fog node through energy optimization and server allocation. However, in both \cite{b3} and \cite{b9}, the inter-cell interference was ignored and offloading decisions and RAN RA were not considered. The authors  in \cite{b10}  optimized the traffic allocation in a multi-tier sliced architecture, while preventing over-provisioning. However RAN and computation RA were considered abstractly, i.e., neither subchannel, power, and computation RA were considered, nor offloading decisions were optimized. Similarly, in \cite{b2}, an abstract view of 'resource' was adopted to minimize the weighted system delay, i.e.,  RAN and computation RA were not addressed.
% 	In contrast to \cite{b1}, in both \cite{b3} and \cite{b9} multiple cooperating ESs were considered. However, in \cite{b3}, offloading decisions were not taken into account, and RAN RA was limited to finding the ratio of available bandwidth for users. 

{Recently, using stochastic optimization, joint subchannel, power and computation RA was considered in a multi-cell sliced network to minimize system energy consumption in \cite{last}, while ignoring offloading decisions}. It should be noted that energy consumption can be modeled as a convex function of transmit power and subchannel allocation variables, and is  different from delay, which at its simplest form, is a function of inverse of non-convex data rate. Also, when all users offload, as in \cite{last}, the delay can be easily restated in terms of the users' data rate. However, with offloading decision optimization, such simplifications are not applicable.
	
% 	Moreover, overlooking offloading decision further simplifies \cite{last} compared to our work. This is due to the fact that when all users supposedly offload, the delay function can easily be restated in term of the data rate. This simplification cannot be employed when offloading decision is taken into account and data rate is an inconsequential factor for many local users.    

% 	Based on this review it can be concluded that RAN resource allocation is completely overlooked in existing literature on slicing. On the other hand works that address this issue do not consider slicing.
% Compared to the existing literature, our contributions can be summarized as:
% \begin{center}
%  \begin{tabular}{||c c c c c ||} 
%  \hline
%  & \cite{b4, b14, b16} & \cite{b5} & \cite{b1,b3, b9, b10, b2} & \cite{b3, b9}  \\ [0.1ex] 
%  \hline\hline
%  Network Slicing & - & - & 787 \\ 
%  \hline
% Interference & - & NOMA & 5415 \\
%  \hline
%   Joint RA & Limited & Limited & 7507 \\
%  \hline
%  Server Cooperation & - & \checkmark & 7507 \\
%  \hline
%  Partial Offloading & - & - & - \\
%   \hline
%  Partial Offloading & - & - & - \\
%  \hline
% \end{tabular}
% \end{center}

{To our best knowledge, the problem of \textit{ delay minimization with joint offloading, computation, and communication RA in a  cooperative multi-cell MEC network with or without slicing} is not investigated in the literature}. Our contributions are:

$\bullet$  We jointly optimize users' offloading decisions, {RAN and computing RA}  in a multi-cell MEC network to  minimize the weighted sum of the difference between the delay observed at each slice and its corresponding desired delay. {The fractional form of the objective function, discrete subchannel allocation, the partial offloading scheme, and the interference incorporated in the rate function, turns this problem into a mixed integer non-linear programming problem (MINLP) for which we proposed an efficient and novel algorithm}.

$\bullet$  We decouple the original problem into two sub-problems: (i)~offloading decision-making and (ii)~joint computation resource, subchannel, and {power allocation}. We solve the first sub-problem optimally. For the second sub-problem, we propose an efficient algorithm with polynomial computational complexity, leveraging on  tools from fractional programming and  Augmented Lagrangian method (ALM). Using alternating optimization, we solve these two sub-problems iteratively until convergence. Complexity analysis is also presented. 

$\bullet$ {Simulation results demonstrate the efficacy of our proposed algorithm compared to existing schemes and provide insights related to the impact of interference, slice prioritization, and cooperative MEC offloading, while demonstrating the convergence in a few iterations.}

% $\bullet$ We investigate the impact of \textit{inter-cell interference} on the delay of slices. To the best of our knowledge, the effect of this factor on this objective function  has not been addressed in existing literature on multi-server MEC networks with offlaoding decision optimization \cite{b4,b5, b14,b16,b1,b3, b9, b10, b2,last}. Note that the algorithms proposed in NOMA-MEC research works, such as \cite{b16}, are specific to NOMA systems and cannot be generalized to scenarios where successive interference cancellation is not employed.

% $\bullet$ We Consider \textit{ES cooperation} while adopting a \textit{partial offloading} scheme. This is unlike existing works in which either server cooperation is ignored \cite{b1,b4,b10,b2,last} or a binary scheme is utilized to address it \cite{b3,b9,b5}.  
\vspace{-1.5mm}
% To further clarify the contributions of our work compared to existing work on network slicing, Table 1 has been provided.
%  \begin{center}
% 	\begin{table}
% 		\centering
% 		\caption{Comparison with existing works on MEC network slicing}
% 		\renewcommand{\arraystretch}{1.35}
% 		\begin{tabular}{|p{1.8cm} |p{1.8cm}|p{1.8cm}|}
% 			\hline
% 			Col1 & Col2 & Col3 \\[2.25pt] \hline\hline
% 			$\sigma^2$ & -120 dBm\\ [2.35pt] \hline
% 			$p_{\max}$ & 42 dBm\\ [2.35pt] \hline
% 			$p_{\max}^u$ & 23 dBm\\ [2.35pt]\hline
% 			$R_\mathrm{min}^\mathrm{dl}$  & 10 bps/Hz\\ [2.35pt]\hline
% 			$R_\mathrm{min}^\mathrm{ul}$  & 5 bps/Hz\\ [2.35pt]\hline
% 			$P_\mathrm{c}^\mathrm{u}$  & 0.1 W\\ [2.35pt]\hline
% 			$P_\mathrm{c}^\mathrm{MBS}$  & 1 W\\ [2.35pt]\hline
% 			$\kappa$  & 38\%\\ [2.35pt]\hline
% 			$\psi$  & 20\% \\ [2.35pt]\hline
% 			$\lambda$  & $10^6$ \\ [2.35pt]\hline
% 			Path loss exponent  &3 \\ [2.35pt]\hline
% 		\end{tabular}
% 	\end{table}
% \end{center}
% % \begin{center}
% %  \begin{tabular}
% % \caption{Your caption.}
% %  \hline
% %  & \cite{b3, b9} &  \cite{b1, b10, b2} \\ [0.1ex] 
% %  \hline\hline
% % Interference & - & -  \\
% %  \hline
% %   Joint RA & - & -  \\
% %  \hline
% %  Server Cooperation &  \checkmark & - \\
% %  \hline
% %  Partial Offloading & - & -  \\
% %   \hline
% %  \hline
% % \end{tabular}
% % \end{center}
\begin{figure}
    \centering
    \includegraphics[width=70mm,height=50mm]{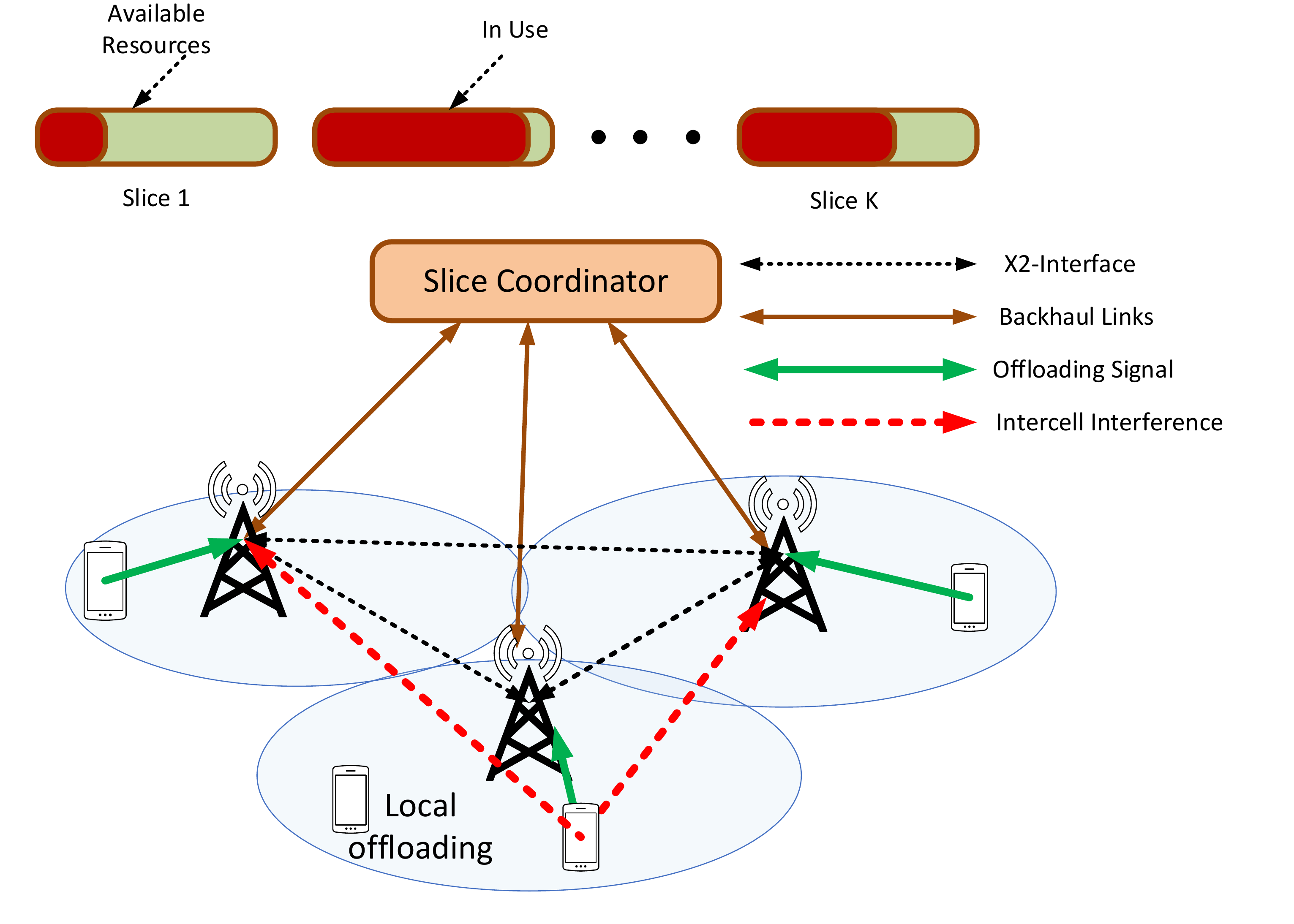}
    \caption{System Model}
    \label{fig:system_model}
    \vspace{-5mm}
\end{figure}
\vspace{-1mm}
\section{System Model and Assumptions}
We consider a MEC network with $M$ edge points (or BSs) with co-located servers\footnote{The edge nodes can connect to each other using any type of topology such as full-mesh or star topology.}. The set of MEC servers is denoted as $
\mathcal{M} = \{1, 2, \cdots, M\}$.  The available spectrum at each cell is divided into $N$ subchannels each with bandwidth $B$.{ Network resources are sliced  to accommodate $\mathcal{K}= \{1,2,\cdots, K\}$ tenants each of which provide one specific type of service}.  Furthermore, the set of users for each tenant $k$ is denoted by $\mathcal{U}_k$ and the set of all users is $\mathcal{U}=\{1,2,\cdots,U\}$. {Each tenant $k$ has a SLA with InP in which the proportion of computation capacity, $\beta_k^E$, and available bandwidth, $\alpha_k$, reserved for its users is determined. The task of each user $u$ is represented by the tuple ($L_u$,$C_u$), with $L_u$ as the size of the task and $C_u$ as the computational demand (CPU cycles) to process each bit.

% Upon receiving a task from a device, each MEC server can either perform the task itself or forward it to other servers with adequate available resources. 	 
To facilitate slice resource management, we consider a software-defined network (SDN) controller referred to as  slice coordinator (SC). 
% When a server receives a task for which it does not have enough resources, it would immediately send a request to SC, stating the requirements of the received task. 
The SC  keeps track of resource utilization in each slice and ensures that service providers (SPs) follow resource constraints in SLA and do not exceed their share of resources. This network architecture is given in Fig. \ref{fig:system_model}. 
% with enough computation capability to service the task in minimum time span while conforming to the resource utilization limit specified in SLA of SPs. 

We denote $y_{u,j}$ as the proportion of the task of user $u$ executed on the MEC server $j$. Thus, we have $ \sum_{j \in \{\mathcal{M}\cup 0\}} y_{u,j} = 1,~\forall u \in \mathcal{U}$, where index $0$ denotes local computation.

\subsubsection{Communication Model}
We consider that if a user offloads its task, it first sends it to its assigned server denoted by $m_u$, and then the remaining communication (possible hand-offs between servers) would be done over the high speed backhaul links. Denoting $\hat{\mathcal{U}}_j$ as the set of users associated to server $j$, the data rate of each user $u$ over subchannel $n$ is:
\begin{flalign}
r_{u,n}\hspace{-1mm}=B\log \left(1+\frac{x_{u,n}~ p_{u,n} h_{u,m_u,n}}{\sigma^2 + I_{u,n}}
\right)
\end{flalign}
where $p_{u,n}$, $I_{u,n}$, and $\sigma^2$  represent the transmit power of user $u$ over subchannel $n$, its inter-cell interference calculated as $I_{u,n}=\sum_{j\in \{\mathcal{M} \backslash m_u\}} \sum_{u'\in \hat{\mathcal{U}}_j}  x_{u',n} p_{u',n} h_{u',m_u,n} $, and receiver noise power, respectively. Also, $h_{u,j,n}$ is the path-gain between user $u$ and BS $j$ over subchannel $n$, and $x_{u,n}$ denotes the binary subchannel allocation variable which is equal to one if subcarrier $n$ is assigned to user $u$, and zero otherwise.

% \begin{flalign}\label{inter} 
% \end{flalign}
Now, we can calculate the total data rate of each user $u$ as $
R_{u}(\mathbf{X},\mathbf{P}) = \sum_{n \in \mathcal{N}} r_{u,n}$,
where $\mathcal{N}$, $\mathbf{X}$, $\mathbf{P}$, denote the set of $N$ subchannels, subchannel allocation matrix, and transmit power allocation matrix, respectively.
Denoting $\mathbf{Y}$ as the matrix of offloading decisions, the communication delay of user $u$ is:
\begin{flalign}
T^{\mathrm{comm}}_{u,m_u}(\mathbf{X},\mathbf{P},\mathbf{Y}) = \frac{\sum_{j \in \mathcal{M}}y_{u,j}L_u}{R_u(\mathbf{X},\mathbf{P})}.
\vspace{-2mm}
\end{flalign}
\subsubsection{Computing Model}
As a partial offloading scheme is adopted here, users' task may be partly processed locally. Denoting the computation capability of local device  for user $u$ as $f_u^L$ {(CPU cycles per second)}, the local computation delay would be:
\begin{flalign}
       T_u^\mathrm{L}(\mathbf{Y})= \frac{y_{u,0} L_u C_u}{f_u^L}.
\end{flalign}
{With $\mathbf{F}$ representing the matrix of all computation resource allocation variables}, since the task of user $u$ might be processed by servers other than its assigned server, the computation delay of user $u$ is:
\begin{flalign}\label{comp_delay}
T_{u}^E (\mathbf{X},&\mathbf{P},\mathbf{F},\mathbf{Y})= T_{u,m_u}^{\mathrm{comm}}(\mathbf{X},\mathbf{P},\mathbf{Y})\nonumber\\&+ \sum_{j \in \{\mathcal{M}\backslash m_u\}}y_{u,j}T^{\mathrm{ho}}_{m_u,j}
+ T^{\mathrm{comp}}_{u}(\mathbf{F},\mathbf{Y}).
\end{flalign}
where $T^{\mathrm{ho}}_{m_u,j}$ denotes the hand-off delay, including the time for communicating with SC and the average round trip time for task transfer between $m_u$ and $j^{\mathrm{th}}$ server. Moreover, $T^{\mathrm{comp}}_{u}$ denotes the offloading computation delay of user $u$. {If tasks' fragments are processed sequentially (one after the other), $T^{\mathrm{comp}}_{u}$ would be the summation of delays of user $u$ in each server $j$ as in (\ref{t_comp}). In case of parallel processing, the computation delay of user would be equal to the delay in the slowest server. However, in order to retain a tractable form for our objective function,  we consider an upper-bound and calculate the computation delay in both cases as follows:}
\begin{equation}\label{t_comp}
T^{\mathrm{comp}}_{u}(\mathbf{F},\mathbf{Y}) = \sum_{j \in \mathcal{M}} \frac{y_{u,j} L_u C_u}{f_{u,j}},
\vspace{-3mm}
\end{equation} 
where $f_{u,j}$ represents the computation resource that is allocated to user $u$ in server $j$ {(CPU cycles per second)}. {Note that even when parallel computation of the tasks is possible , due to 1) positivity of computation delay and 2) the independence between $f_{u,j}$ for {different servers}, this upper bound would not significantly effect the optimized value of computation resource allocation in the slowest server, as minimizing the sum translates into minimizing each component separately.} Due to the typically small size of response, we ignore the downlink transmission delay. Thus, the total delay of each user $u$ is:\vspace{-2mm}
\begin{flalign}
&T_u(\mathbf{X},\mathbf{P},\mathbf{F},\mathbf{Y}) =
T_u^L(\mathbf{Y})+ T^E_{u}(\mathbf{X},\mathbf{P},\mathbf{F},\mathbf{Y}).
\end{flalign}
%where, $T_c$, is the time required for combining the result obtained from MEC server with that of local computation. We should note that if local computation is faster than offloading, then after finishing the local processing, user must wait until the result of edge processing is returned from the edge server. Also, if the response from edge computation is received before local computation is over, the total delay a user should spend on a task would be equal to the delay of local computation. Regarding the aforementioned scenarios it is clear that taking overlap of  edge and local processing time is vital in accurately estimating the task completion time. This fact is why instead of summation of local and offloading delays, we use a maximum operator on them. 
 
%Each slice in the network has a maximum delay tolerance that we denote by $\Bar{T_k}~ \forall k \in \mathcal{K}$. Although the goal is to satisfy all these requirements, in reality in many occasions due to conjestion and high traffic load it becomes impossible not to violate these QoS requirements. With regard to this fact, we consider that we have some high priority services that are more restrict in their QoS requirement whereas other types of services can tolerate deviation from their maximum delay tolerance threshold. In this regard, we divide the set of slices into two subsets: the set of slices with high priority services, $\mathcal{K}^{Pr}$ and the set of slices with normal services $\mathcal{K}^{No}$, where $\mathcal{K} = \mathcal{K}^{Pr} \cup \mathcal{K}^{No}$.
\vspace{-5mm}
\section{Problem Formulation}
In this section, we formulate the problem of minimizing the weighted sum of the difference between the delay observed at a given slice and its corresponding delay requirement (\textit{or weighted sum of the delay deviation at each slice}), through jointly optimizing users' offloading decisions, {RAN and computing RA}  in a cooperative multi-cell MEC network. 
This problem offers SPs a valuable insight into the adequacy of their leased resources to meet the service quality requirement of their subscribers and the average delay they would experience under the existing SLA. Analysing the results obtained, SPs can better plan their future strategies to whether maintain their current SLA, invest more on leasing resources, or to modify their subscription policy to either increase or decrease the number of users they accept.
Now, we formally state the optimization problem as follows:
\vspace{-7mm}
\begin{center}
	\begin{equation}\label{first}
	\begin{aligned}
	&\textbf{P}:\underset{\mathbf{X},\mathbf{P},\mathbf{F},\mathbf{Y}}{\text{min}}~ \sum_{k \in \mathcal{K}} \sum_{u\in \mathcal{U}_k} \lambda_k (T_u(\mathbf{X},\mathbf{P},\mathbf{F},\mathbf{Y})-\Bar{T_k}) \\
	& \text{Subject to:}\\
	& C_1:\sum_{u\in \hat{\mathcal{U}_j}} x_{u,n} \leq 1, \qquad\qquad \forall n\in \mathcal{N}, \forall j \in \mathcal{M},\\
	& C_2:x_{u,n} \in \{1,0\}, \qquad  \forall n\in \mathcal{N}, \forall j \in \mathcal{M},\forall u \in \hat{\mathcal{U}_j},\\
	& {C_3: 0 \leq \sum_{n \in \mathcal{N}} x_{u,n} p_{u,n} \leq P_{\mathrm{max},u} , \qquad \forall{u}\in \mathcal{U},}\\
 	& C_4:  y_{u,0} L_u C_u \leq F^L_{u}, \qquad\qquad \forall u\in \mathcal{U},\\
 	& C_5: \sum_{u \in \mathcal{U}} y_{u,j}L_u C_u \leq F^E_{j}, \qquad \forall j\in \mathcal{M},\\
& C_6: \sum_{u\in \mathcal{U}_k} \sum_{n \in \mathcal{N}} x_{u,n} \leq  \alpha_k M N , \quad  \forall{k} \in \mathcal{K},\\
& C_7: \sum_{u \in \mathcal{U}_k} \sum_{j \in \mathcal{M}} f_{u,j}\leq \beta_k^E S^E, \qquad \forall k \in \mathcal{K},\\
& C_8: \sum_{j \in \{\mathcal{M}\cup 0\}} y_{u,j}  = 1, \qquad\qquad \forall u \in \mathcal{U},\\ 
& C_9:y_{u,j} \in [0,1], \quad \qquad \qquad \forall u \in \mathcal{U}, \forall j \in \mathcal{M}.
	\end{aligned}
	\end{equation}
\end{center}
In the above optimization problem, $\Bar{T_k}$ denotes the desired delay threshold of each slice $k$ and {$\lambda_k$ is the weighting factor whose value is defined in SLA} and handles the precedence of slices over each other. Furthermore, constraint $C_1$ indicates that each subchannel can be allocated to at most one user in each cell and $C_2$ shows the binary nature of the subchannel allocation variable. In constraint $C_3$, users' transmit power {is restricted between zero and a maximum threshold denoted by $P_{\mathrm{max},u}$}. In constraints $C_4$ and $C_5$, the limitation of local and edge computation resources are specified for each user and server, respectively, with $F_u^L$ and $F_j^E$ denoting the total computation capacity of user $u$ and server $j$ {(both in CPU cycles per second)}, in that order. Constraints $C_6$ and $C_7$ ensure that resource consumption at each slice follows SLA. That is,  $C_6$ limits the spectrum usage for each slice. Since there are $M$ cells in the system and each cell has access to $N$ subchannels, then in total we have $NM$ subchannels, from which only $\alpha_k$ percent can be used by users of slice $k$. Similar to communication resources, the proportion of the total computation capacity {$S^E$ ($S^E = \sum_{j\in \mathcal{M}}F^E_j$)} that is allocated to each slice $k$ is limited to $\beta_k^E$ as given in constraint $C_7$. Constraints $C_8$ and $C_9$ clarify the partial offloading decision scheme adopted in this work.

% {\em Remark 1:} For $\bar{T}_k=0$, the problem becomes a weighted delay minimization problem.
% In (\ref{first}), the first two constraints relate to subchannel allocation. 

As the result of interference included in the rate function, the binary subchannel allocation variables, and the objective function which is in the form of summation of ratios, optimization problem (\ref{first}) is MINLP and thus difficult to tackle. In the what follows we present our resource allocation algorithm.
\vspace{-2mm}
\section{Proposed Resource Allocation Framework}
To tackle the difficulties of solving problem (\ref{first}), we first take advantage of the problem structure and decompose it into the following two subproblems:
\vspace{-6mm}
\begin{center}
	\begin{equation}\label{p1}
	\begin{aligned}
	\textbf{P1}:&\underset{\mathbf{Y}}{\text{min}}~ \sum_{k \in \mathcal{K}} \sum_{u\in \mathcal{U}_k} \lambda_k (T_u(\mathbf{Y})-\Bar{T_k}) \\
	& \text{Subject to:}~C_4,~C_5,~C_8,~C_9.
	\end{aligned}
	\end{equation}
\end{center}
\vspace{-4mm}
\begin{center}
	\begin{equation}\label{p2}
	\begin{aligned}
	\textbf{P2}:&\underset{\mathbf{X},\mathbf{P},\mathbf{F}}{\text{min}}~ \sum_{k \in \mathcal{K}} \sum_{u\in \mathcal{U}_k} \lambda_k (T_u(\mathbf{X},\mathbf{P},\mathbf{F})-\Bar{T_k}) \\
	& \text{Subject to:}~C_1-C_3,~C_6,~C_7.
	\end{aligned}
	\end{equation}
\end{center}
In problem (\ref{p1}), both the objective function and constraint set are affine with respect to the variable $\mathbf{Y}$. As such, it can be  solved using standard optimization tools such as CVX toolbox.

The first challenge in (\ref{p2}) is the multiplication of subchannel and power allocation variables in (1) as well as in constraint $C_3$. To tackle this challenge, we first replace all $x_{u,n} p_{u,n}$ terms with $p_{u,n}$ and then add the following constraint to (\ref{p2}):
\vspace{-6mm}
\begin{center}
	\begin{equation}
	\begin{aligned}
	{C_{3,1}:0 \leq p_{u,n}\leq x_{u,n} P_{max,u}}
	\end{aligned}
	\end{equation}
\end{center}
\vspace{-1mm}
By using the above modification, users' transmit power would be automatically set to zero over subchannels they do not own. By adding this constraint, data rate function $R_u(\mathbf{X},\mathbf{P})$ would become a function of trasmit power only ($R_u(\mathbf{P})$). This step solves the variable multiplication issue, however discrete subchannel allocation variable is still challenging. To deal with this issue we replace $C_2$ with the following two constraints:
\vspace{-1mm}
\begin{align}
	C_{2,1}: 0 \leq x_{u,n} \leq 1,~~ C_{2,2}: x_{u,n}-x_{u,n}^2 \leq 0.
	\vspace{-1mm}
\end{align}

{\em Remark 2:} Although we relax $x_{u,n}$ to a continuous variable in $C_{2,1}$, since the only two values in [0,1] that fit $C_{2,2}$ are 0 and 1, the binary nature of this variable would be preserved.

 The fractional form of users' delay, $T_u$, is the next issue we focus on. After offloading decision is obtained through solving subproblem $\mathbf{P1}$, edge computation delay, $T^{\mathrm{comp}}_{u}(\mathbf{F},\mathbf{Y})$, in the objective function of  $\mathbf{P2}$ would turn into a convex function  and hand-off delay would be a constant. This leaves us with the summation of users' transmission delay, whose non-convexity can be easily  proved. 
 \begin{lemm} Using tools from fractional programming, problem (\ref{p2}) can be restated as:
 \vspace{-8mm}
\begin{center}
	\begin{equation}\label{p3}
\begin{aligned}
\underset{\mathbf{X},\mathbf{P},\mathbf{F}}{\text{min}}~ &T(\mathbf{P},\mathbf{F}) =\sum_{k \in \mathcal{K}} \sum_{u\in \mathcal{U}_k} \lambda_k \Big[{y_{u,j} L_u}\frac{1}{{2R_u(\mathbf{P})}^2}\\&+\hspace{-4mm}\sum_{j \in \mathcal{M},j\ne m_u}\hspace{-3mm}y_{u,j}T^{\mathrm{ho}}_{m_u,j}+\sum_{j \in \mathcal{M}}\frac{y_{u,j} L_u C_u}{f_{u,j}}-\Bar{T_k}\Big] \\
	& \text{Subject to:}~C_1-C_3,~C_6,~C_7.
	\end{aligned}
	\end{equation}
\end{center}
\end{lemm}
\begin{proof}
~An optimization problem with the form 
$\underset{x\in C_x}{\min} \sum_{i=1}^I \frac{B_i({X})}{A_i({X})}$,
can be restated equivalently  as \cite{wei}:
\vspace{-3mm}
\begin{equation}\label{1}
\underset{x\in C_x,t\in \mathbb{R}^+ }{\min} \sum_{i=1}^I t_i {B_i(X)}^2+\sum_{i=1}^I \frac{1}{4t_i}\frac{1}{{A_i(X)}^2},
\end{equation}
where $t_i=\frac{1}{2B_i(X)A_i(X)}$. 
Using \eqref{1} and by setting $B_i= 1$ and $A_i={R_u(\mathbf{P})}$, we restate problem (\ref{p2}) as given in Lemma~1.
\vspace{-2mm}
\end{proof}

Due to the presence of interference, $R_i(\mathbf{P})$ is still a non-convex function of transmit power. 
\begin{lemm} We can obtain an equal but convex representation of communication delay function by restating the rate  as:
 \vspace{-6mm}
\begin{center}
	\begin{equation}\label{rate_app}
	\small
\begin{aligned}
\hat{r}_{u,n}(\mathbf{P},z_{u,n})=\hspace{-1mm}\log_2\Big(\hspace{-0.5mm}1+2z_{u,n}\sqrt{h_{u,m_u,n}p_{u,n}}-z^2_{u,n}(I_{u,n}+\sigma^2)\Big),
	\end{aligned}
	\end{equation}
\end{center}
\vspace{-1mm}
\end{lemm}
\begin{proof}
~ As mathematically proven in \cite{rate} and since in \textbf{Lemma 1}, we set $A_i=R_u(\textbf{P})$, and $R_u(\textbf{P})= \sum_{n\in \mathcal{N}} r_{u,n}$, $r_{u,n}$ can be equally restated as (\ref{rate_app}). This modification, not only makes $r_{u,n}$ a concave function of $\mathbf{P}$, $\frac{1}{R_u(\textbf{P})^2}$ would also become a convex function.
\end{proof}
\vspace{-2mm}
\normalsize
In (\ref{rate_app}), $z_{u,n}$ is a slack variable that will be updated iteratively. Using \textbf{Lemma 2}, we convexify the complex non-convex function $T_{u,m_u}^{\mathrm{comm}}(\mathbf{P})$, also we redefine $R_u(\mathbf{P}) = \sum_{n \in \mathcal{N}} \hat{r}_{u,n}(\mathbf{P},z_{u,n})$. For optimizing $\mathbf{P}$, $\mathbf{X}$, and $\mathbf{F}$ we adopt ALM. For a given $z_{u,n}$, the augmented Lagrangian function is given in (\ref{AM}).
\begin{figure*}
\begin{align}\label{AM}
&\min L(\mathbf{X},\mathbf{P},\mathbf{F},\mathbf{Z},\mathbf{\Gamma})= T(\mathbf{X},\mathbf{P},\mathbf{F},\mathbf{Z})+\frac{1}{2\Psi}\Bigg[\Bigg(\Bigg[\sum_{u \in \mathcal{U}}\theta_u+\Psi(\sum_{n \in \mathcal{N}} p_{u,n}\nonumber-P_{max,u}) \Bigg]^+\Bigg)^2-\sum_{u \in \mathcal{U}}\theta^{2}_u\nonumber\\&+\Bigg(\Bigg[\sum_{k \in \mathcal{K}} \delta_k+\Psi\big(\sum_{u \in \mathcal{U}_k} \sum_{j\in \mathcal{M}} f_{u,j}-\beta_k^E S^E\big)\Bigg]^+\Bigg)^2-\sum_{k \in \mathcal{K}} \delta_k^2+
\Bigg(\Bigg[\sum_{n \in \mathcal{N}} \sum_{j \in \mathcal{M}} \phi_{n,j}+\Psi\big(\sum_{u \in \mathcal{U}}x_{u,n}-1\big)\Bigg]^+\Bigg)^2\nonumber\\
&-\sum_{n \in \mathcal{N}} \sum_{j\in \mathcal{J}} \phi_{n,j}^2+\Big(\Bigg[\sum_{n \in \mathcal{N}} \sum_{j \in \mathcal{M}} \sum_{u \in \mathcal{U}} \xi_{u,n,j}+\Psi \bigg(\sum_{n \in \mathcal{N}} \sum_{j \in \mathcal{M}} \sum_{u \in \hat{\mathcal{U}}_j} x_{u,n}-x_{u,n}^2\bigg)\Bigg]^+\Big)^2-\sum_{n \in \mathcal{N}} \sum_{j \in \mathcal{M}} \sum_{u \in \mathcal{U}} \xi_{u,n,j}^2\nonumber\\&+\Bigg(\Bigg[\sum_{u \in \mathcal{U}}\sum_{n \in \mathcal{N}}\Xi_{u,n}+\Psi(\sum_{u \in \mathcal{U}}\sum_{n \in \mathcal{N}} p_{u,n}-x_{u,n}P_{max,u}) \Bigg]^+\Bigg)^2-\sum_{u \in \mathcal{U}}\sum_{n \in \mathcal{N}}\Xi^{2}_{u,n}\Bigg].
\end{align}	
\hrule
\end{figure*}

In the augmented Lagrangian function, $\Psi$ is a positive constant that plays the role of an adjustable penalty coefficient and $\Gamma$ is the vector of all Lagrangian multipliers $\Theta$, $\Delta$, $\Phi$, $\xi$, and $\Xi$. Solving problem {(\ref{p2}) or, equivalently (12)} can be done in three steps. In the first step, we consider Lagrangian multipliers to be fixed and minimize $L(\mathbf{X},\mathbf{P},\mathbf{F},\mathbf{Z})$ given in (\ref{AM}). In the second step, Lagrangian multipliers would be updated as:
\vspace{-3mm}
\begin{align}
\theta_u^{t+1}=\Bigg[\theta_u^{t}+\Psi\Bigg(\sum_{n \in \mathcal{N}} p_{u,n}-P_{max,u}\Bigg)\Bigg]^+,
\end{align}
\vspace{-4.5mm}
\begin{align}
\delta_k^{t+1}=\Bigg[\delta_k^{t}+\Psi\big(\sum_{u \in \mathcal{U}_k} \sum_{j\in \mathcal{M}} f_{u,j}-\beta_k^E S^E\big)\Bigg]^+,
\end{align}
\vspace{-4.5mm}
\begin{align}
\phi^{t+1}_{n,j}=\Bigg[\phi^{t}_{n,j}+\Psi\big(\sum_{u \in \mathcal{U}}x_{u,n}-1\big)\Bigg]^+,
\end{align}
\vspace{-4.5mm}
\begin{align}
\xi^{t+1}_{u,n,j}\hspace{-2mm}=\hspace{-1.5mm}\Bigg[\hspace{-1mm}\sum_{n \in \mathcal{N}} \sum_{j \in \mathcal{M}} \sum_{u \in \mathcal{U}}\hspace{-0.8mm}\xi^{t}_{u,n,j}\hspace{-1.5mm}+\Psi \bigg(\hspace{-0.5mm}\sum_{n \in \mathcal{N}} \sum_{j \in \mathcal{M}} \hspace{-0.6mm}\sum_{u \in \hat{\mathcal{U}}_j}\hspace{-1.5mm}x_{u,n}-x_{u,n}^2\bigg)\Bigg]^+
\vspace{-4.5mm}
\end{align}
\vspace{-4.5mm}
\begin{align}
\Xi_{u,n}^{t+1}=\hspace{-1.5mm}\Bigg[\sum_{u \in \mathcal{U}}\sum_{n \in \mathcal{N}}\Xi^{t}_{u,n}+\Psi(\sum_{u \in \mathcal{U}}\sum_{n \in \mathcal{N}} p_{u,n}\hspace{-0.5mm}-x_{u,n}P_{max,u}) \Bigg]^+.
\end{align}
The third step is executed after a solution is obtained for (\ref{AM}). In this last step, using the values obtained for $\mathbf{P}$ and $\mathbf{X}$, we update slack variable $z_{u,n}$ as
$z_{u,n}=\frac{\sqrt{p^*_{u,n}h_{u,m_u,n}}}{(I_{u,n}+\sigma^2)}$. Our proposed algorithm is given in \textbf{Algorithm 1}.
\vspace{-1mm}
		\begin{algorithm}
		\caption{ Proposed Algorithm}
		\label{euclid}
		\begin{algorithmic}[1]
			\State Obtain the solution of problem (\ref{p1}) and initialize $Z$.
			\State\textbf{Repeat} 
			\State~~Initialize $\Gamma=[ \Theta, \Delta, \Phi, \xi, \Xi]$ with small numbers. 
			\State~~\textbf{Repeat}
			\State~~~~Solve problem (\ref{AM}) considering $\Gamma$ to be fixed,
			\State~~~~Update $\Gamma$ using (16), (17), (18), (19), and (20).
			\State~~\textbf{Until} convergence.
			\State~~~Set $z_{u,n}=\frac{\sqrt{p^*_{u,n}h_{u,m_u,n}}}{(I_{u,n}+\sigma^2)}$ for all users and subchannels.
			\State\textbf{Until} Convergence
		\end{algorithmic}
	\end{algorithm}
\vspace{-2mm}
\section{Computation Complexity Analysis}
Our proposed algorithm is divided into two sub-problems, i.e., (i) offloading decision optimization and (ii)~joint computation and RAN RA. {For the first sub-problem, we use interior point method in CVX whose complexity is in the order of  $O(\log(\frac{C/t_{0}\xi}{\epsilon}))$, where $C$, $t_{0}$, $\xi$, and $\epsilon$ denote the total number of constraints, the initial point for interior point method, the stopping criterion, and  a representation of the accuracy of the method, respectively}. For the second sub-problem based on ALM the order of complexity at each iteration is $\mathcal{O}(KUM)^2$ which is  polynomial \cite{ata}.
\vspace{-4mm}
\section{Simulation Results and Discussions}
\begin{figure*}
\vspace{-2.5mm}
\centering
\begin{minipage}{0.3\linewidth}
\includegraphics[width=6.2cm, height=5.2cm]{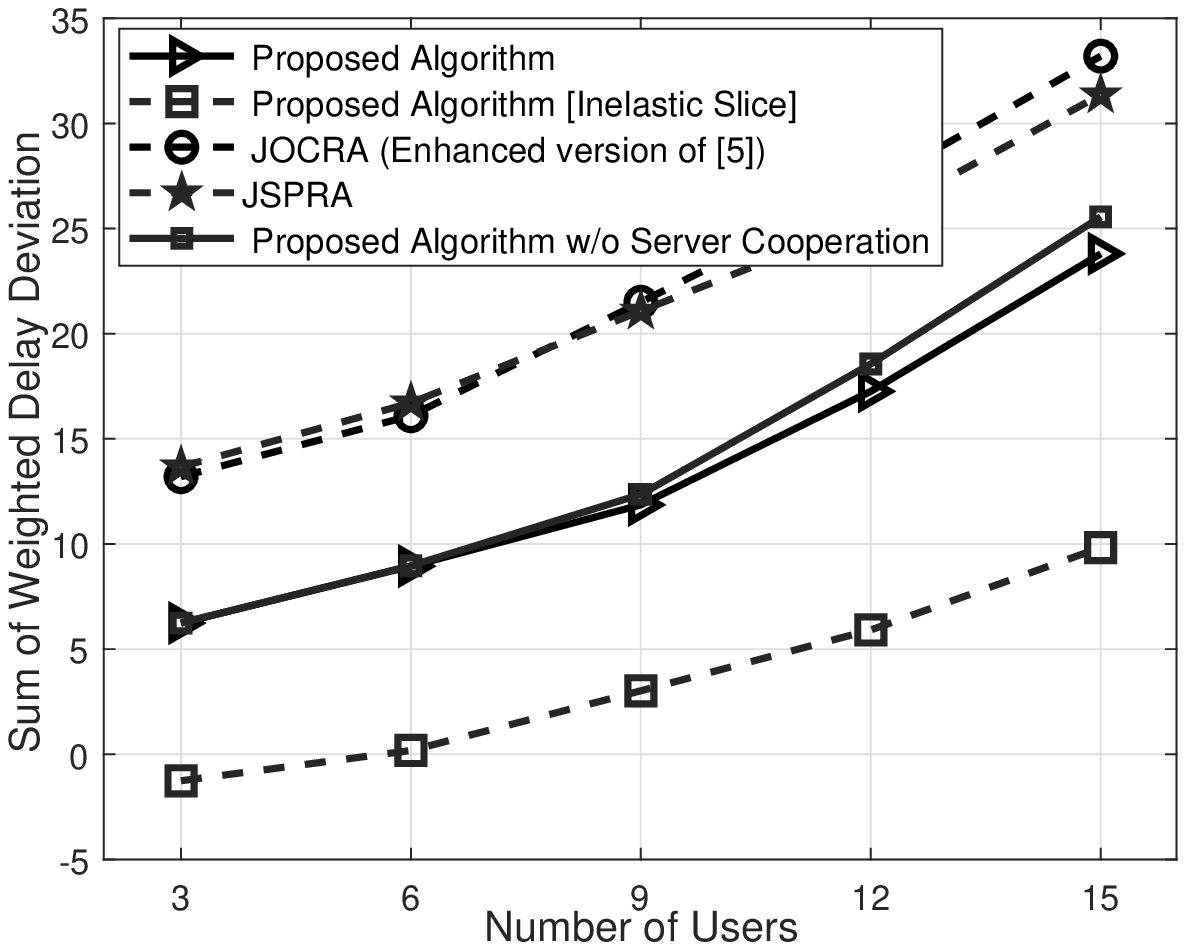}
\caption{Weighted delay deviation Vs. Number of users}
\end{minipage}\hfill
\begin{minipage}{0.3\linewidth}
\includegraphics[width=6.2cm, height=5.2cm]{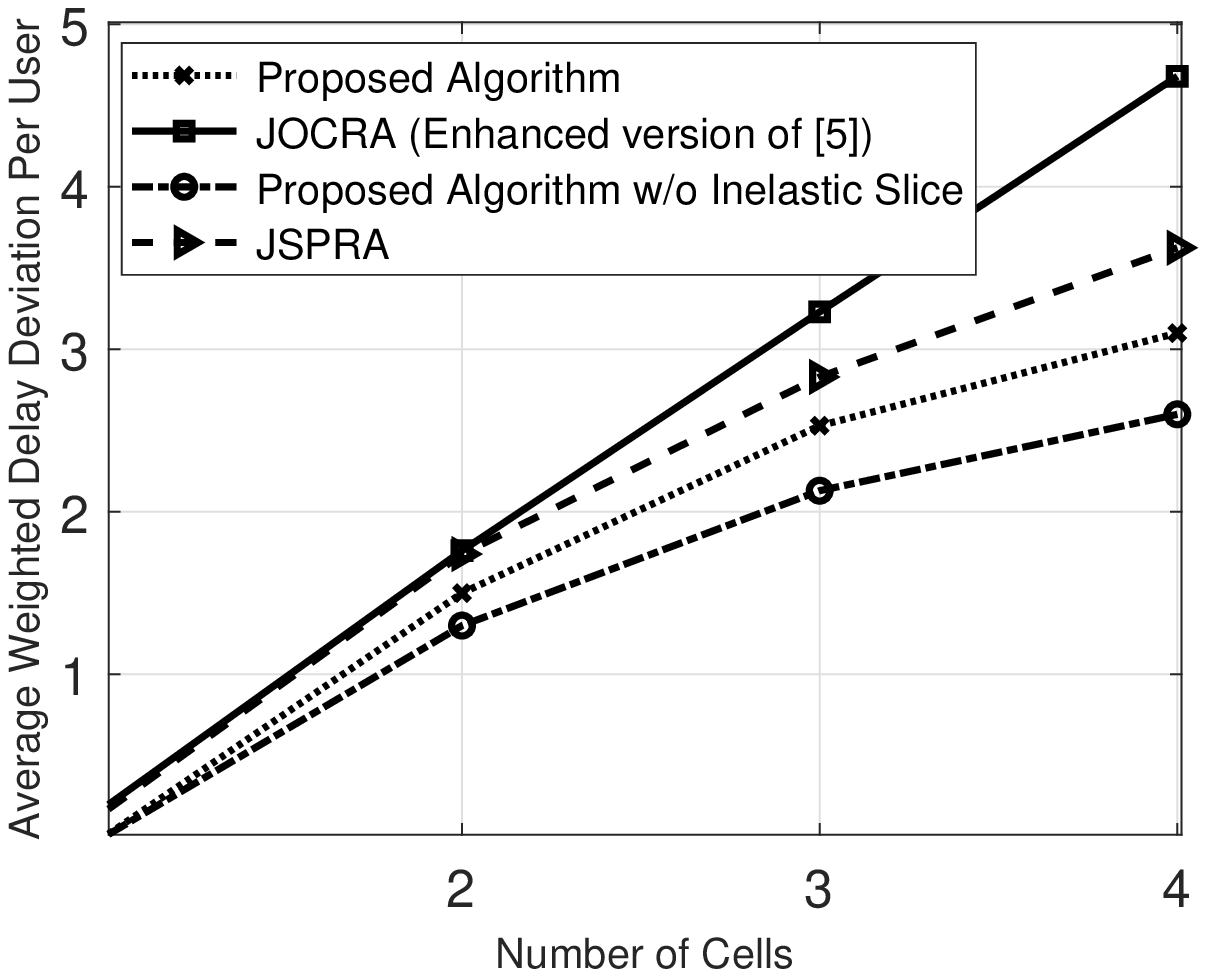}
\caption{Weighted delay deviation Vs. Number of cells}
\end{minipage}\hfill
\begin{minipage}{0.3\linewidth}
\includegraphics[width=6.1cm, height=5.2cm]{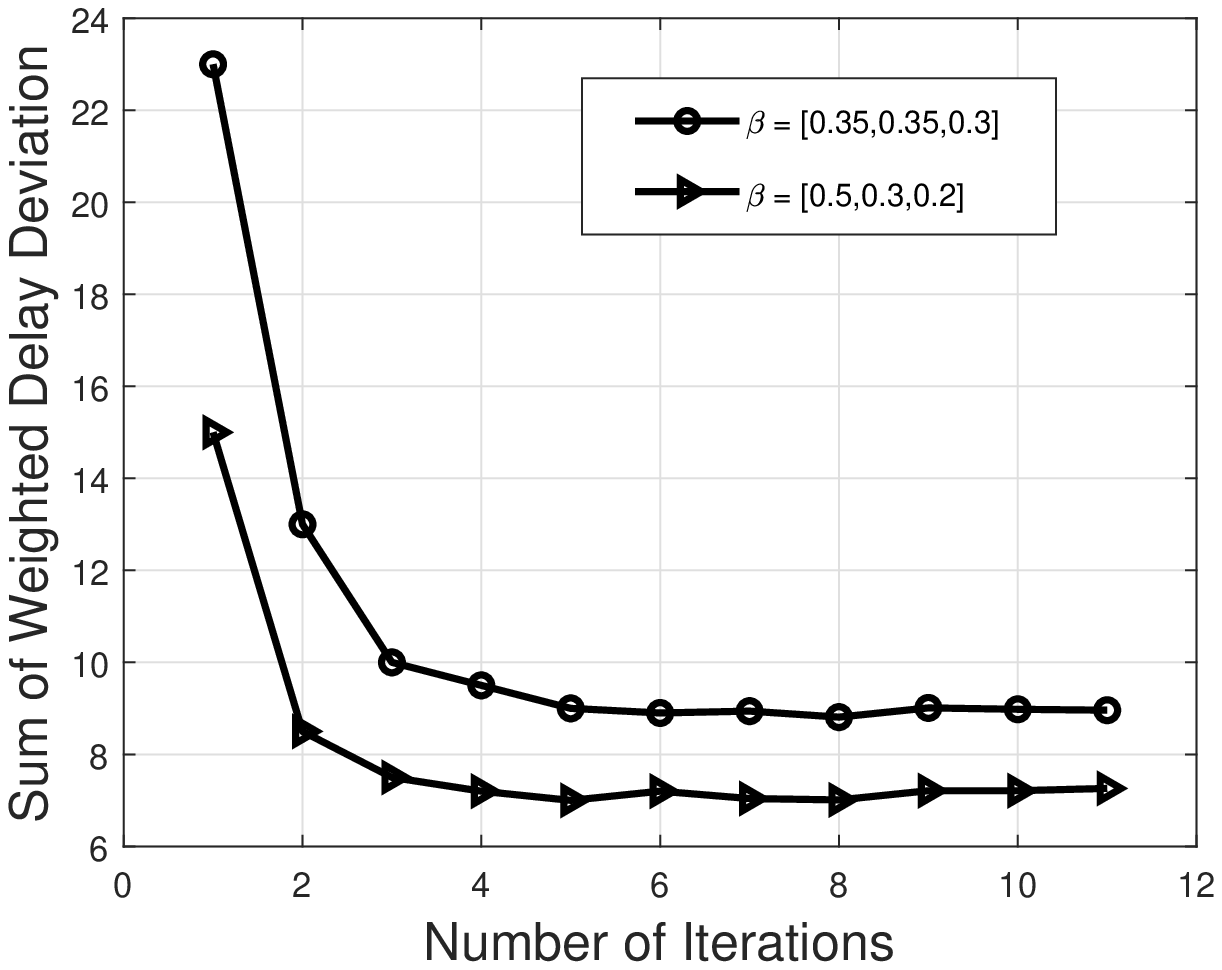}
\caption{Weighted delay deviation Vs. Number of iterations}
\end{minipage}
\label{fig2}
\vspace{-5.5mm}
\end{figure*}

% In this section, we evaluate our proposed algorithm which considers {\em joint offlaoding decision, communication and computation RA with server cooperation} and compare it with the  schemes, i.e., \textbf{i)} { Joint RA with No Cooperation}, \textbf{(ii)} 
% Joint RA; however, subchannels are allocated randomly and maximum transmit power is used,
% \textbf{(iii)} Joint RA; however,  computation resources are equally allocated to users.
% Also, we investigate how different factors such as joint communication and computation RA, inter-cell interference, and cooperation can impact the users' delay. 

We consider a network with two cells each having 6 users and 16 available subchannels, unless stated otherwise. Similar to \cite{b2}, {we consider three slices/services as: \textit{elastic services} with flexible latency constraints, \textit{inelastic services} that require ultra-low latency, and \textit{background services} with low latency requirement}. The weighting parameter $\lambda$ is set to  $[3,2,1]$ for inelastic, elastic, and background services, with 50ms, 100ms, and 5s desired delay threshold, respectively. The value of $L_u$ is 1 MB and the CPU cycle, $C_u$, is randomly chosen from $[1500, 2000, 2500]$. {As a convex problem, initial point does not effect the solution of (8), however, to avoid increasing the complexity, the initial point of the problem (9) is obtained by  checking various values and selecting the best values that minimizes our objective function}.

Fig.~2 depicts the effect of number of users in each cell on the sum of weighted delay deviation at each slice. We have  compared our algorithm  with 1) \textit{Joint Offloading and Computation RA (JOCRA):} where only offloading and computing RA is considered (with interference and server cooperation, this scenario is in fact an improvement on \cite{b3}), 2) \textit{Joint offloading, Subchannel, Power RA (JSPRA):} in which only RAN RA is addressed and computation resource is equally allocated to users, and 3) our proposed scheme without server cooperation. We can clearly observe the significance of joint computation and RAN RA in the delay that users experience. In fact, if we ignore computation RA we would have $58\%$ and if we overlook communication RA we will have $62\%$ increase in network delay deviation on average.  In Fig. 1, the impact of cooperation among cells is also illustrated. At first, when number of users is not too high, there is almost no need for cooperation. However, as the number of users increases, we observe that the effect of cooperation becomes  noteworthy (i.e., $9\%$ reduction on average). {The positive delay deviation occur when network becomes infeasible (i.e., insufficient resources in at least one slice) and satisfying the QoS of high priority services takes precedence in the network. Thus, we can  preserve the QoS of slices by increasing their weight ($\lambda_k$) for prioritization of the slice or the quota of reserved resources ($\beta$ and $\alpha$) to avoid infeasibility. However, such modifications are often a function of the cost SPs are willing to pay.}
% This would also provide a good evaluation point for SPs to see if the resources they have leased are enough to satisfy their users' demands.
% \begin{figure}[h]
% \vspace{-5mm}
% \centering
% \includegraphics[width=8cm, height=5.3cm]{images/Nuser_2.eps}
% \caption{Weighted delay deviation Vs. Number of users}
% \end{figure}
% \begin{figure}[h]
% \centering
% \vspace{-4mm}
% \includegraphics[width=8cm, height=5.35cm]{images/Number_of_cells.eps}
% \caption{Weighted delay deviation Vs. Number of cells}
% \end{figure}
% \begin{figure}[h]
% \centering
% \vspace{-4mm}
% \includegraphics[width=7cm, height=4.5cm]{images/convergence.eps}
% \caption{Weighted delay deviation Vs. Number of iterations}
% \end{figure}

In Fig.~3, we examine how increasing the number of cells impacts the delay of users. We again compare our proposed algorithm with \textit{JOCRA} and \textit{JSPRA}. As the number of users per cell remains constant here, we depict the average delay deviation per user. Increasing the number of cells notably increases the delay of users, however this increase is more significant when communication  RA is overlooked. Because, {while the average amount of resources available for users remains almost the same (since the number of users in each cell is constant), more cells means intensified interference in the network. To deal with the negative effect of this intensified interference, precise RAN RA becomes imperative.} 
% We also illustrate the delay deviation when inelastic services are removed from the system.

{The convergence of our proposed algorithm and the importance of slice resource management is numerically demonstrated in Fig. 4. Here, we observe that: i) our algorithm converges to its final solution after a few iterations, and ii) careful resource reservation plays a significant role in the QoS users of each slice achieve.}
\vspace{-4mm}
\section{Conclusion}
In this work we propose a framework to minimize the delay in cooperative MEC network by optimizing both \textit{RAN and computation} resources and \textit{offloading decisions}, using  tools from fractional programming, convexification of rate function, and ALM. {The problem of routing between edge servers is a venue for future works, especially with wireless backhauling.}
% Analysing the results of this optimization problem can provide a deeper insights to service providers on the adequacy of their leased resources to meet their service quality and  helps them better plan their future investment strategies to whether maintain their current SLA, invest more on leasing resources, or to modify their subscription policy.

\vspace{-3mm}
\bibliographystyle{IEEEtran}
\bibliography{ref}
\end{document}